\documentclass[12pt,reqno,a4paper]{amsart}

\usepackage{amsthm,amsmath,amssymb,amsfonts}
\usepackage[utf8]{inputenc}
\usepackage[doi=false,isbn=false,style=alphabetic,sorting=nyt,backend=biber,maxnames=5,maxalphanames=5,firstinits=true]{biblatex}
\bibliography{interaction}

\usepackage[linesnumbered,lined,commentsnumbered]{algorithm2e}
\usepackage{caption}
\usepackage{subcaption}

\usepackage{enumerate}
\usepackage{multirow}
\usepackage{array}
\usepackage{stmaryrd}
\usepackage{enumitem}
\usepackage{centernot}
\usepackage{todonotes}
\usepackage{mathrsfs}
\usepackage{listings}
\usepackage{multicol}
\usepackage{color}
\usepackage[noend]{algorithmic} 

\usepackage{array}
\PassOptionsToPackage{hyphens}{url}
\usepackage[colorlinks,plainpages,hypertexnames=false,plainpages=false,hyperindex,breaklinks]{hyperref}
\hypersetup{urlcolor=blue, citecolor=blue, linkcolor=blue}
\usepackage{upgreek}
\usepackage{amsmath}
\usepackage{amsthm}
\usepackage{amssymb}
\usepackage{amssymb}
\usepackage{amsfonts}
\usepackage{enumerate}
\usepackage{stmaryrd}
\usepackage{enumitem}
\usepackage{centernot}
\usepackage{todonotes}
\usepackage{mathrsfs}
\usepackage{listings}
\usepackage{multicol}
\usepackage{lscape}
\usepackage{color}
\usepackage[noend]{algorithmic} 

\usepackage{array}
\newcolumntype{L}[1]{>{\raggedright\let\newline\\\arraybackslash\hspace{0pt}}m{#1}}
\newcolumntype{C}[1]{>{\centering\let\newline\\\arraybackslash\hspace{0pt}}m{#1}}
\newcolumntype{R}[1]{>{\raggedleft\let\newline\\\arraybackslash\hspace{0pt}}m{#1}}
\usepackage{tikz}
\usetikzlibrary{arrows}
\usetikzlibrary{decorations.pathreplacing}
\usetikzlibrary{matrix}
\usetikzlibrary{calc}
\usetikzlibrary{shapes}
\usetikzlibrary{patterns}
\usetikzlibrary{fit,backgrounds,scopes}

\newtheoremstyle{theoremstyle}
{10pt}      
{5pt}       
{\itshape}  
{}          
{\bfseries} 
{}          
{ }         
{}          

\newtheoremstyle{algorithmstyle}
{10pt}      
{5pt}       
{}          
{}          
{\bfseries} 
{}          
{ }         
{}          

\newtheoremstyle{examplestyle}
{10pt}      
{5pt}       
{}          
{}          
{\bfseries} 
{}          
{ }         
{}          

\theoremstyle{theoremstyle}
\newtheorem{theorem}{Theorem}[section]
\newtheorem*{theorem*}{Theorem}

\newtheorem{proposition}[theorem]{Proposition}
\newtheorem*{proposition*}{Proposition}

\newtheorem*{corollary*}{Corollary}

\theoremstyle{examplestyle}
\newtheorem{example}[theorem]{Example}
\newtheorem{definition}[theorem]{Definition}
\newtheorem{definition*}{Definition}

\newtheorem{remark*}{Remark}

\newcommand{\RR }{\mathbb{R}}

\newcommand{\W}{{\rm W}}
\newcommand{\suchthat}{\;\ifnum\currentgrouptype=16 \middle\fi|\;}
\newcommand{\bigmid}{\left.\vphantom{\Big\{} \suchthat \vphantom{\Big\}}\right.}

\newcommand\irregularcircle[2]{
	\pgfextra {\pgfmathsetmacro\len{(#1)+rand*(#2)}}
	+(0:\len pt)
	\foreach \a in {10,20,...,350}{
		\pgfextra {\pgfmathsetmacro\len{(#1)+rand*(#2)}}
		-- +(\a:\len pt)
	} -- cycle
}
\DeclareMathOperator*{\mini}{minimize}
\DeclareMathOperator*{\st}{subject~to}

\date{}
\title[Cooperativity, Absolute Interaction, and Algebraic Optimization]{Cooperativity, Absolute Interaction,\\ and Algebraic Optimization}
\author{Nidhi Kaihnsa}
\address{Max Planck Institute for Mathematics in the Sciences\\
  Inselstra{\ss}e 22\\
  04103 Leipzig\\
  Germany
}
\email{nidhi.kaihnsa@mis.mpg.de}
\urladdr{https://personal-homepages.mis.mpg.de/kaihnsa/}

\author{Yue Ren}
\address{Max Planck Institute for Mathematics in the Sciences\\
  Inselstra{\ss}e 22\\
  04103 Leipzig\\
  Germany
}
\email{yue.ren@mis.mpg.de}
\urladdr{https://yueren.de}

\author{Mohab Safey El Din}
\address{Sorbonne Universit\' e, CNRS, INRIA, Laboratoire d’Informatique de Paris 6, LIP6, \'Equipe \textsc{PolSys}
}
\email{mohab.safey@lip6.fr}
\urladdr{https://www-polsys.lip6.fr/~safey/}

\author{Johannes W. R. Martini}
\address{KWS SAAT SE\footnote{this research project is not associated with KWS SAAT SE}\\
  Grimsehlstraße 31\\
  37574 Einbeck\\
  Germany.
}
\email{jwrmartini@gmail.com}

\begin{document}

\maketitle

\begin{abstract} 
We consider a measure of cooperativity based on the minimal absolute interaction required to generate an observed titration behavior.
We describe the corresponding algebraic optimization problem and show how it can be solved using the nonlinear algebra tool \texttt{SCIP}.
Moreover, we compute the minimal absolute interactions for various binding polynomials that describe the oxygen binding of various hemoglobins under different conditions.
While calculated minimal absolute interactions are consistent with the expected outcome of the chemical modifications, it ranks the cooperativity of the molecules differently than the maximal Hill slope.
\end{abstract}

\section{Introduction}
Interaction between components is a key concept of biological systems.
While characteristics of a system of independent subunits are determined by the description of each element alone, a complex system with interactions is more than the sum of its parts. A classical example of a small biological system with non-trivial interaction is hemoglobin with its four binding sites for oxygen or other ligands such as carbon monoxide \cite{bohr1904ueber,barcroft1913combinations,hill1913combinations}. The overall binding curve, as a function of oxygen concentration at constant temperature (``isotherm''), is
much steeper than an overall binding curve generated by a system of independent sites can be. Hemoglobin ``prefers'' the extreme states of being fully (un)occupied to intermediate states with one, two or three ligands bound.
This dependent binding behavior or an ``abnormal" steepness of a response curve is described by the concept of ``cooperativity''. Cooperativity is not only considered to be important for
(oxygen) transport, but for instance also for the formation of multi-protein complexes \cite{roy2017cooperative}, general signal transduction processes \cite{salakhieva2016kinetic,lenaerts2009information,martini2016model} and the regulation of noise \cite{gutierrez2009role,monteoliva2013noise}.

Observations on the macroscopic level, that is on how many sites are occupied have often been connected to microscopic behavior that is which sites are occupied and how certain sites interact (for a review see \cite{stefan2013cooperative}).
However, it has been shown that in the framework of the grand canonical ensemble of statistical mechanics, different definitions of cooperative binding are not coinciding when general asymmetric systems with more than two binding sites are considered \cite{martini2016cooperative,martini2017relation}. In particular, macroscopic phenomena which have been attributed to positive cooperativity in literature can be caused by a system being negative cooperative on the microcopic level \cite{martini2016cooperative}. 
\\

 For the characterization of the macroscopic binding behavior, roots of the grand partition function (also called the ``binding polynomial'') have been the objects of investigations \cite{briggs1983new,briggs1984cooperativity,briggs1985relationship}. In particular, it has been pointed out that the criterion of having non-real roots is  generally the relevant aspect indicating that a binding process relies on non-negligible interaction \cite{martini2016cooperative}. In case that the binding polynomial has only real roots, these roots define hypothetical single-site-systems whose sum give the initial overall curve \cite{onufriev2001novel,onufriev2004decomposing,martini2013mathematical}. The macroscopic behavior of the system could thus theoretically be engineered by using these independent sites which suggests that the potentially present intrinsic interaction is not essential for the macroscopic functionality of the system. Conversely, in case that the binding polynomial has non-real roots, the overall system cannot be represented as a decoupled system consisting of independent ``real'' binding sites suggesting that the intrinsic interaction is essential for the macroscopic behavior of the system. This qualitative definition raises the question of how to quantify the cooperativity of a system.  \\

The quantification of cooperativity has predominantly been discussed for symmetric systems consisting of physically indistinguishable sites \cite{abeliovich2016hill,rong2018quantification,rong19}. In particular, in the framework of the grand canonical ensemble, the molecule consisting of identical sites can be easily inferred from the coefficients of the binding polynomial and the corresponding interaction energies can be determined \cite{martini2016cooperative}. Moreover, it has been proposed to use the absolute minimal interaction which is required to generate the overall binding behavior as a general measure of cooperativity when the binding sites are not considereed to be physically identical \cite{martini2017measure}. In this general situation without additional side conditions, an observed overall titration behavior can be explained by an infinite number of different hypothetical molecules. A central question to characterize cooperativity can then be what the minimal absolute interaction across all these molecules is and in particular whether it is zero.
This quantitative measure of cooperativity has been defined and upper bounds based on polynomial factorization or on the interaction of a symmetric system, have been derived \cite{martini2017measure}.
However, it has not been clear how to explicitly calculate this measure of cooperativity for larger systems. \\

In this work, we show that looking for the minimal interaction energy required to generate a given overall titration behavior, leads to an algebraic optimization problem which can be tackled using the software {\tt SCIP} \cite{SCIP}. We describe the mathematical problem and calculate examples for binding polynomials from native hemoglobins and some modified variants \cite{connelly1986analysis,ikeda1983thermodynamic,imai1973analyses}. For our examples, the minimal interaction reduces when the molecule structure of hemoglobin is disturbed.
However, our results rank the degree of cooperativity of chemically treated hemoglobines differently than the maximal Hill slope.


\section{Recapitulation of the mathematical framework}\label{sectionmathframework}

In this section, we briefly review the model for ligand binding based on the grand canonical ensemble of statistical mechanics and recall the notion of minimal absolute interaction from \cite{martini2017measure}. We refer the reader to  \cite{ben2013cooperativity,wyman1990binding,hill2013cooperativity} for detailed exposition on the model.

\begin{definition}[Molecule]
  A \emph{molecule} $\W$ with $n$ sites is a positive real point, whose coordinates are indexed by the subsets of $[n]:=\{1,\ldots,n\}$ and with $w_\emptyset=1${\rm :}
  \[ \W = (w_I)_{I\subseteq [n]} \in (\mathbb \RR_{>0})^{2^n}. \]
  We refer to $w_I$ as \emph{binding energies} if $|I|=1$ and \emph{interaction energies} if $|I|>1$.
  For the sake of brevity, we abbreviate $w_{\{i_1,\ldots,i_r\}}$ to $w_{i_1\ldots i_r}$ for $i_1<\dots<i_r$ (see Fig.~\ref{fig:moleculeFourSites}).
\end{definition}

The binding energy $w_i$ encodes the difference between free energy of the
molecule in the completely unoccupied state and the state with only a ligand
bound to site $i$. The interaction energy $w_{i_1 i_2}$ is the
difference between free energy when sites $i_1,i_2$ are occupied and the sum of
the differences of energy encoded by $w_{i_1}$ and $w_{i_2}$. Analogously, $w_{i_1 \dots i_r}$ encodes the discrepancy between the energy required for $r$ ligands to bind to the sites $i_1,\dots,i_r$ and the energy already defined by the interactions of lower degree.
Note that we may call $w_{i_1 \dots i_r}$ energies even though they actually represent exponentials of the physical quantity ``energy'' \cite{ben2013cooperativity,wyman1990binding,schellman1975macromolecular}.



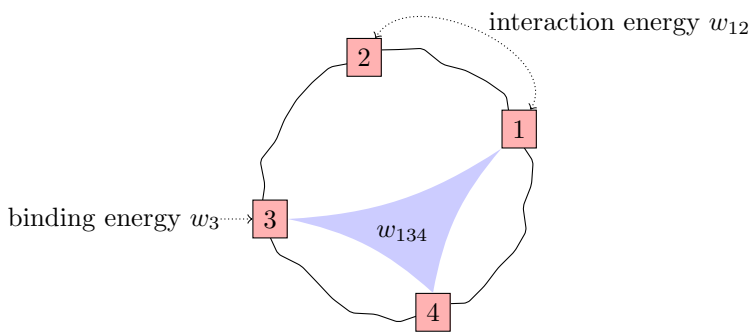
\begin{figure}[h]
  \centering
  \begin{tikzpicture}[scale=1,every node/.style={font=\footnotesize}]
    \pgfmathsetseed{13371338}
    \coordinate (o) at (0,0);
    \draw[rounded corners=.5mm] (o) \irregularcircle{1.75cm}{1mm};
    \draw (25:1.75cm) node[draw,fill=red!30] (s1) {1};
    \draw (105:1.75cm) node[draw,fill=red!30] (s2) {2};
    \draw (195:1.75cm) node[draw,fill=red!30] (s3) {3};
    \draw (285:1.75cm) node[draw,fill=red!30] (s4) {4};

    \node[anchor=east,xshift=-0.5cm] (aText) at (s3) {binding energy $w_3$};
    \draw[->,densely dotted] ($(aText.east)+(-0.15,0)$) -- (s3);
    \draw[<->,densely dotted] (s1) to[out=60,in=60] node[anchor=south west,xshift=0cm,yshift=-2mm] {interaction energy $w_{12}$} (s2);
    \path[fill=blue!20] (s3.east) to[bend left=15] (s4.north) to[bend left=15] (s1.south west) to[bend left=15] (s3.east);
    \node[anchor=south,xshift=-4mm,yshift=8mm] at (s4) {$w_{134}$};
  \end{tikzpicture}
  \caption{A molecule with 4 binding sites.}
  \label{fig:moleculeFourSites}
\end{figure}



\begin{definition}[Binding Polynomial]
  Given a molecule $\W=(w_I)_{I\subseteq [n]}$ with $n$ sites, its {binding polynomial} $\Phi(\W)$ is a univariate polynomial of degree $n$,
  \begin{equation*}\label{bindingpol}
   \Phi(\W):=a_n \Lambda^n + \dots + a_1 \Lambda + a_0,
   \end{equation*}
  whose positive real coefficients $a_k$ are given by
  \begin{equation*}
   \label{eq:bindingPolynomialCoefficients}
    a_k := \sum_{|I|=k} \prod_{I'\subseteq I} w_{I'} \in \RR_{>0}.
  \end{equation*}
\end{definition}

\begin{example}\label{ex:background}\rm
  Let $\W=(w_I)_{I\subseteq[3]}\in (\RR_{>0})^{2^3}$ be a molecule with $3$ sites. The binding polynomial $\Phi(\W)=a_3\Lambda^3+a_2\Lambda^2+a_1\Lambda+a_0$ is a real univariate polynomial of degree~$3$ whose coefficients are given by
  \begin{align*}
    a_0 &= w_\emptyset = 1,\\
    a_1 &= w_1+w_2+w_3,\\
    a_2 &= w_1w_2w_{12}+w_1w_3w_{13}+w_2w_3w_{23},\\
    a_3 &= w_1w_2w_3w_{12}w_{13}w_{23}w_{123}.
  \end{align*}
  Note that two different molecules may have the same binding polynomial and thus the same binding curve. Hence, the map $\Phi$ which maps a molecule with $n$ sites to its binding polynomial is not injective.
\end{example}

%



\begin{definition}[Absolute interaction]\label{def:norm}
  The \emph{absolute interaction} of a molecule $\W=(w_I)\in (\RR_{>0})^{2^n}$ is given by
  \begin{equation*}
    \| \W \| := \prod_{|I|>1} \max\Big(w_I,w_I^{-1}\Big).
  \end{equation*}
\end{definition}

Since $w_I$ represents the exponential of a difference of actual energies, which can be positive or negative, $\max(w_I,w_I^{-1})$ represents the exponential of the absolute value of that difference. Hence, the absolute interaction $\|\W\|$ represents the exponential of the sum of the absolute values of all energy differences.
The minimal value $\| \W \| $ can adopt is $1$ which corresponds to a minimal physical interaction energy of $0$. 

\begin{definition}[Minimal absolute interaction]
  Given a binding polynomial $P$ of degree $n$, the \emph{minimal absolute interaction} is
  \begin{equation*}
    \|P\| := \min \Big\{ \| \W \| \bigmid \W\in (\RR_{>0})^{2^n} \text{ with } \Phi(\W)=P \Big\}.
  \end{equation*}
  We call a molecule $W$ \emph{minimal}, if $\|W\| = \|\Phi(W)\|$.
\end{definition}

In other words, the minimal absolute interaction $\|P\|$ is the minimum of all absolute interactions of molecules with a fixed binding polynomial $P$, or equivalently with a fixed binding curve.
This notion is well defined by the following proposition.

\begin{proposition}[{\cite[\S4]{martini2017measure}}]
  For any univariate polynomial $P$ of degree $n$ with positive coefficients there exists a molecule $\W\in (\RR_{>0})^{2^n}$ with $n$ sites such that $\Phi(\W)=P$ and $\|\W\|=\|P\|$.
\end{proposition}

Since cooperativity is an emerging property of the interaction between the sites and deeply connected to the binding curve, the minimal absolute interaction of a binding polynomial is a natural candidate for quantifying cooperativity. In addition, \cite[\S4]{martini2017measure} shows that it has several properties that could be considered desirable in such a quantifier.

\begin{example}\label{example3site}\rm
Consider the following binding polynomials:
  \begin{align*}
    P_1:=&4\Lambda^3+3\Lambda^2+2\Lambda+1,\\
    P_2:=&6\Lambda^3+7\Lambda^2+4\Lambda+1=(2\Lambda+1)(3\Lambda^2+2\Lambda+1),\\
    P_3:=&6\Lambda^3+11\Lambda^2+6\Lambda+1=(\Lambda+1)(2\Lambda+1)(3\Lambda+1).
  \end{align*}
  A brief computation reveals that the minimal absolute interactions are
  \[ \|P_1\|=13.50,\qquad \|P_2\|=3, \qquad \|P_3\|=1. \]
  The computation for $\|P_1\|$ is explicitly shown in Example~\ref{example3dim}. Figure~\ref{fig:minimalMolecules} illustrates the minimal molecules for each binding polynomial. Since $P_3$ factorizes into three real linear factors, its minimal molecule has only trivial interactions.
\end{example}

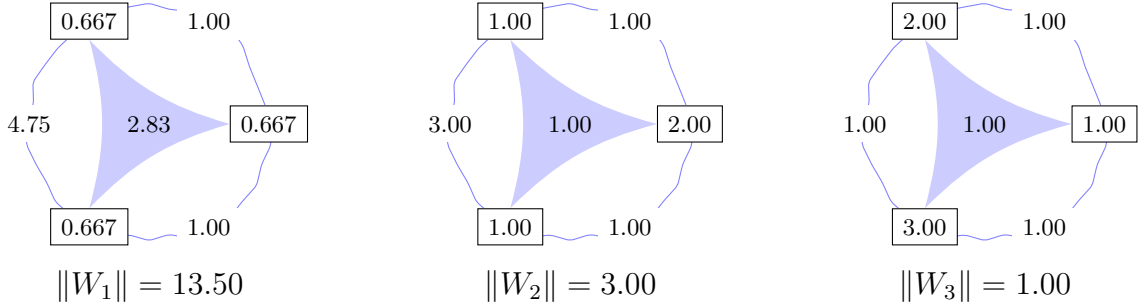
\begin{figure}[h]
  \centering
  \begin{tikzpicture}[scale=0.9]
    \pgfmathsetseed{42424242}
    \coordinate (o) at (0,0);
    \draw[blue!50,rounded corners=.5mm] (o) \irregularcircle{1.75cm}{1mm};
    \draw (0:1.75cm) node[draw,fill=white,font=\scriptsize] (s1) {$0.667$};
    \draw (120:1.75cm) node[draw,fill=white,font=\scriptsize] (s2) {$0.667$};
    \draw (240:1.75cm) node[draw,fill=white,font=\scriptsize] (s3) {$0.667$};

    \node[fill=white,font=\scriptsize] (w12) at (60:1.75cm) {$1.00$};
    \node[fill=white,font=\scriptsize] (w23) at (180:1.75cm) {$4.75$};
    \node[fill=white,font=\scriptsize] (w13) at (300:1.75cm) {$1.00$};

    \path[fill=blue!20] (s3.north) to[bend left=15] (s1.west) to[bend left=15] (s2.south) to[bend left=15] (s3.north);
    \node[font=\scriptsize] (w123) {$2.83$};

    \node[anchor=north] at (0,-2) {$\|W_1\|=13.50$};
  \end{tikzpicture}\hfill
  \begin{tikzpicture}[scale=0.9]
    \pgfmathsetseed{42424242}
    \coordinate (o) at (0,0);
    \draw[blue!50,rounded corners=.5mm] (o) \irregularcircle{1.75cm}{1mm};
    \draw (0:1.75cm) node[draw,fill=white,font=\scriptsize] (s1) {$2.00$};
    \draw (120:1.75cm) node[draw,fill=white,font=\scriptsize] (s2) {$1.00$};
    \draw (240:1.75cm) node[draw,fill=white,font=\scriptsize] (s3) {$1.00$};

    \node[fill=white,font=\scriptsize] (w12) at (60:1.75cm) {$1.00$};
    \node[fill=white,font=\scriptsize] (w23) at (180:1.75cm) {$3.00$};
    \node[fill=white,font=\scriptsize] (w13) at (300:1.75cm) {$1.00$};

    \path[fill=blue!20] (s3.north) to[bend left=15] (s1.west) to[bend left=15] (s2.south) to[bend left=15] (s3.north);
    \node[font=\scriptsize] (w123) {$1.00$};

    \node[anchor=north] at (0,-2) {$\|W_2\|=3.00$};
  \end{tikzpicture}\hfill
  \begin{tikzpicture}[scale=0.9]
    \pgfmathsetseed{42424242}
    \coordinate (o) at (0,0);
    \draw[blue!50,rounded corners=.5mm] (o) \irregularcircle{1.75cm}{1mm};
    \draw (0:1.75cm) node[draw,fill=white,font=\scriptsize] (s1) {$1.00$};
    \draw (120:1.75cm) node[draw,fill=white,font=\scriptsize] (s2) {$2.00$};
    \draw (240:1.75cm) node[draw,fill=white,font=\scriptsize] (s3) {$3.00$};

    \node[fill=white,font=\scriptsize] (w12) at (60:1.75cm) {$1.00$};
    \node[fill=white,font=\scriptsize] (w23) at (180:1.75cm) {$1.00$};
    \node[fill=white,font=\scriptsize] (w13) at (300:1.75cm) {$1.00$};

    \path[fill=blue!20] (s3.north) to[bend left=15] (s1.west) to[bend left=15] (s2.south) to[bend left=15] (s3.north);
    \node[font=\scriptsize] (w123) {$1.00$};

    \node[anchor=north] at (0,-2) {$\|W_3\|=1.00$};
  \end{tikzpicture}  \caption{Minimal molecules for binding polynomial $P_1$, $P_2$, $P_3$.}
  \label{fig:minimalMolecules}
\end{figure}

\section{The Algebraic Optimization Problem}\label{sectionalgopt}
In this section, we consider the computation of the minimal absolute interaction as an optimization problem with the absolute interaction as objective function and the set of all molecules sharing the same binding polynomial as feasible set:
\begin{equation}\label{orgopt}
  \begin{array}{lll}
    \displaystyle\mini_{\W} & \quad & \| \W \|\\[3mm]
    \st        & & \Phi(\W)=P
  \end{array}
\end{equation}
This problem seems simple and its concept can be understood easily. However, it
becomes quickly complicated if larger systems are considered. For small $n$ it
can be worked out by methods of algebraic optimization (e.g. using sums of squares
approaches \cite{Lasserre01, Parrilo03} or critical point methods \cite{SG14}).
For $n=4$ explicitly,
given the binding polynomial $P=a_4\Lambda^4+a_3\Lambda^3+a_2\Lambda^2+a_1\Lambda$ we have:
\begin{equation*}
  \begin{array}{lcl}
    \displaystyle\mini_{\W}
    && \displaystyle\prod_{|I|>1} \max\Big(w_I,w_I^{-1}\Big)\\[3mm]
    \st
    && a_1 = w_1+w_2+w_3+w_4,\\[1mm]
    && a_2 = w_1w_2w_{12}+w_1w_3w_{13}+w_1w_4w_{14}\\
    &&\qquad\qquad +w_2w_3w_{23}+w_2w_4w_{24}+w_3w_4w_{34},\\[1mm]
    && a_3 = w_1w_2w_3w_{12}w_{13}w_{23}w_{123}+w_1w_2w_4w_{12}w_{14}w_{24}w_{124}\\
    && \qquad\qquad +w_1w_3w_4w_{13}w_{14}w_{34}w_{134}+w_2w_3w_4w_{23}w_{24}w_{34}w_{234},\\[1mm]
    && a_4 = w_1w_2w_3w_4w_{12}w_{13}w_{14}w_{23}w_{24}w_{34}w_{123}w_{124}w_{134}w_{234}w_{1234}.
  \end{array}
\end{equation*}

One intuitive way to compute the minimal absolute interaction $\|P\|$ is to decompose the space of molecules $(\RR_{>0})^{2^n}$. We look at $2^{2^n-n-1}$ regions on which  for $|I|>1$ either $w_I\geq 1$ or $w_I\leq 1$ holds.
Explicitly, for each set of interactions $\mathcal I\subseteq \{I\subseteq [n]\mid |I|>1\}$, we define a region
\[ O_{\mathcal I}:= \big\{(w_I)_{I \subseteq [n]}\mid \forall |I|>1: w_I\geq 1 \;\forall I\notin\mathcal I \text{ and } w_I\leq 1 \; \forall I\in\mathcal I \big\} \]
so that
\[ (\RR_{>0})^{2^n} = \bigcup_{\mathcal I\subseteq \{I\subseteq [n]\mid |I|>1\}} O_{\mathcal I}. \]
On each region $O_{\mathcal I}$, the absolute interaction $\|W\|$ is a rational function in the interaction energies which we denote by $f_{\mathcal I}(W)$,
\begin{equation}
  \label{eq:fI}
  \|W\| = \prod_{I\in\mathcal I} w_I^{-1} \prod_{\substack{|I|>1\\I\not\in\mathcal I}} w_I=: f_{\mathcal I}(W) \quad \text{for } \W \in \mathcal{O}_{\mathcal{I}}.
\end{equation}
Finding the minimal absolute interaction inside $O_{\mathcal I}$ is a straight-forward problem in polynomial optimization:
\begin{equation}
  \label{eq:OptInO}
  \begin{array}{lll}
    \displaystyle\mini_{\W\in\mathcal O_{\mathcal I}} & \quad & f_{\mathcal I}(\W)\\[3mm]
    \st        & & \Phi(\W)=P.
  \end{array}
\end{equation}

However, this approach requires solving $2^{2^n-n-1}$ polynomial optimization problems, one for each region $\mathcal O_{\mathcal I}$, which is not easily feasible. \\

We can reduce the number of regions that require consideration by exploiting the
intrinsic symmetry of the problem.
To do this, we need to do a small journey in algebra and use the notion of group.
In mathematics, a (commutative) group is a set $S$ equipped with a binary
operation $\diamond$ that combines any two elements to form a third element of
that set and such that this operation is associative (and commutative), the set
contains an identity element ($\exists e \in S$ s.t. for all $x \in S$, $e
\diamond x = x$) and any element of $S$ is invertible (for all $x \in S$, there
exists $y\in S$ s.t. $x\diamond y = e$).
The basic example for groups is the set of integers $\mathbb{Z}$ equipped with the binary operation of addition.

In our setting, we want to exploit symmetry. We will consider the symmetric
group $S_n$ which is the set of all bijections from $\{1, \ldots, n\}$ to $\{1,
\ldots, n\}$ equipped with composition of maps. In the sequel, we use the following
standard notation: $(a b)$, for $a,b$ in $\{1, \ldots, n\}$, denotes the
bijection that permutes $a$ and $b$ in $\{1, \ldots, n\}$.

There exists a natural action of the symmetric group $S_n$ on the space of molecules $(\RR_{>0})^{2^n}$ which permutes the binding sites and the corresponding interaction energies,
\[ S_n\times (\RR_{>0})^{2^n} \longrightarrow (\RR_{>0})^{2^n},\quad (\sigma,(w_I)) \longmapsto \sigma\cdot (w_I) := (w_{\sigma(I)}). \]
This action permutes the regions $\mathcal O_{\mathcal I}$ and preserves binding polynomials as well as absolute interactions:

\begin{proposition}
  The binding polynomial and the minimal absolute interaction of a molecule are invariant under the group action of $S_n$ i.e., $\Phi(\sigma\cdot \W)=\Phi(\W)$ and $\|\sigma\cdot \W\| = \|\W\|$ for all $\W\in (\RR_{>0})^{2^n}$ and all $\sigma \in S_n$.
\end{proposition}

For example for any $1\leq i<j\leq n$, the permutation $(1i)(2j)\in S_n$ induces a bijection between the regions $O_{\{12\}}$ and $O_{\{ij\}}$ and hence
\begin{align*}
&\min\{ \|\W\| \mid \W \in O_{\{12\}} \text{ with }\Phi(\W)=P\} \\
&\qquad\qquad \qquad= \min\{ \|\W\| \mid \W\in O_{\{ij\}} \text{ with }\Phi(\W)=P\}.
\end{align*}

Therefore, it suffices to consider one representative per $S_n$-orbit of
regions. However, the number of required regions remains infeasibly large: We
can identify each region $\mathcal O_{\mathcal I}$ with a hypergraph with
vertices $[n]$ and edges $\mathcal I$ (a hypergraph is a generalization of a
graph in which an edge can join any number of vertices). The number of required
regions then coincides with the number of such hypergraphs up to isomorphism.
Included in the hypergraphs of the regions are all graphs with vertices $[n]$,
and the number of graphs on $n$ vertices up to isomorphism is well known to grow
faster than any exponential function \cite[\S 6.4]{GYZ14}.

Due to this problem, we preprocess the optimization problem. Note that we can rewrite Problem \eqref{orgopt} in the following minimax form, as on each region $\mathcal{O}_{\mathcal{I}}$ the function $f_{\mathcal I}(W)$, as defined in~\eqref{eq:fI}, dominates all other $f_{\mathcal I'}(W)$ with $\mathcal{I}' \neq \mathcal{I}$:
\begin{equation*}
  \begin{aligned}
    &\mini_{W} \qquad \max_{\mathcal I\subseteq \{I\subseteq [n]\mid |I|>1\}} f_{\mathcal I}(W)  \\
    &\st \qquad \quad \Phi(W)=P.
  \end{aligned}
\end{equation*}
This can then be lifted to a problem with linear objective function and non-linear constraints:
\begin{equation}\label{newoptprob}
  \begin{aligned}
    &\mini_{W}\qquad  r  \\
    &\st \quad   \begin{cases} \quad f_{\mathcal I}(W)\leq r \quad \text{for all } \mathcal I\subseteq \{I\subseteq [n]\mid |I|>1\},\\ \quad \Phi(W)=P. \end{cases}   \\
  \end{aligned}
\end{equation}
Given a fixed binding polynomial $P$, the polynomial system $\Phi(W)=P$ poses a serious computational problem. To sidestep this issue we now introduce new coordinates $s_I:=\prod_{I'\subseteq I}w_{I'}$. Here, $s_I$ represents the free energy difference of microstate $I$ to the fully unoccupied state. On the positive orthant $(\RR_{>0})^{2^n}$, this defines a bijection

\begin{flushright}
  \begin{tikzpicture}
    \node (topLeft) {$(\RR_{>0})^{2^n}$};
    \node[anchor=base west,xshift=35mm] (topRight) at (topLeft.base east) {$(\RR_{>0})^{2^n}$};

    \node[anchor=base,yshift=15mm] (midLeft) at (topLeft.base) {$(w_I)$};
    \node[anchor=base east,xshift=2mm] at (midLeft.base west) {$W=$};
    \node[anchor=base,yshift=15mm] (midRight) at (topRight.base) {$\Big(\displaystyle\prod_{I'\subseteq I} w_{I'}\Big)$};

    \node[anchor=base,yshift=-15mm] (bottomRight) at (topRight.base) {$(s_I)$};
    \node[anchor=base west,xshift=-2mm] at (bottomRight.base east) {$=S$};
    \node[anchor=base,yshift=-15mm] (bottomLeft) at (topLeft.base) {$\Big(\displaystyle\prod_{I'\subseteq I} s_{I'}^{(-1)^{|I\setminus I'|}}\Big)$};

    \draw[<->,shorten >=7mm,shorten <=7mm] (topLeft) -- (topRight);
    \draw[|->] ($(topLeft)+(1.5,1.5)$) -- node[below] {$\varphi$} ($(topRight)+(-1.5,1.5)$);
    \draw[<-|] ($(topLeft)+(1.5,-1.5)$) -- node[above] {$\phantom{{}^{-1}}\varphi^{-1}$} ($(topRight)+(-1.5,-1.5)$);

    \draw[draw opacity=0]
    (midLeft) -- node[sloped] {$\in$} (topLeft)
    (bottomLeft) -- node[sloped] {$\in$} (topLeft)
    (midRight) -- node[sloped] {$\in$} (topRight)
    (bottomRight) -- node[sloped] {$\in$} (topRight);

    \node[anchor=base west,text width=35mm,align=right,yshift=7mm] at (topRight.base east) {};
  \end{tikzpicture}\label{eq:s_I}
\end{flushright}

In the new coordinates, the formerly polynomial constraints $a_k = \sum_{|I|=k} \prod_{I'\subseteq I} w_{I'}$ are simplified to linear constraints $a_k=\sum_{|I|=k} s_I$ for $k=1,\ldots,n$. However, the functions $f_{\mathcal I}$ become more complicated. For example, for $n=4$ and $\mathcal I=\{ \{1,2,3\} \}$
\[ f_{\mathcal I} (W) = w_{123}^{-1}\cdot \prod_{\substack{|I|>1,\\ I\neq\{1,2,3\}}} w_I \quad \text{becomes} \quad
  f_{\mathcal I}\circ\varphi^{-1} (S) = \frac{s_{1234}s_{12}^2s_{13}^2s_{23}^2}{s_{123}^2s_1^3s_2^3s_3^3s_4}. \]
Since the $f_{\mathcal I}\circ\varphi^{-1} (S)$ remain monomials (with possibly negative exponents), this complication is of little consequence for the resulting problem:
\begin{equation}\label{eq:optInS}
  \begin{aligned}
    &\mini_{S}\qquad  r  \\
    &\st \quad   \begin{cases} \quad f_{\mathcal I}\circ\varphi^{-1}(S)\leq r \quad \text{for all } \mathcal I\subseteq \{I\subseteq [n]\mid |I|>1\},\\ \quad \sum_{|I|=k} s_I= a_k \quad \text{ for all } k=1,\ldots,n. \end{cases}   \\
  \end{aligned}
\end{equation}

\begin{example}\label{example3dim}
  Consider the polynomial $P_1$ of Example~\ref{example3site}. We use {\tt SCIP}~\cite{SCIP}, which is currently one of the fastest solvers for non-linear programming, to solve the resulting polynomial optimization problem \eqref{eq:optInS}. It uses a branch and bound method to solve the optimization problem with non-linear constraints.

  Figure~\ref{scipformat} shows the full input on the left and the partial output on the right. In the input, the first constraints \texttt{c1}, \texttt{c2}, \texttt{c3} enforce $\sum_{|I|=k} s_I= a_k$ for $i=1,2,3$. The remaining constraints {\tt c4} to {\tt c19} enforce $f_{\mathcal{I}}\circ\varphi^{-1}(S)\leq r$. For example, \texttt{c19} states that $s_1s_2s_3 / s_{123}\leq r$ which is equivalent to $f_{\emptyset}(W)=(w_{123}w_{12}w_{13}w_{23})^{-1} \leq r$ in the coordinates $w_I$. The output states an approximate optimal value of $y=13.5$ and clearly lists all values of $s_I$ of the minimal molecules, which gave rise to the values of $w_I$ in Example~\ref{example3site}.
\end{example}

\begin{figure}
  \begin{minipage}[t]{0.475\linewidth}
\begin{lstlisting}[showlines=true, basicstyle=\scriptsize\ttfamily]
Minimize
obj: r
Subject to
c1: s1+s2+s3=2
c2: s12+s13+s23=3
c3: s123=4
c4: s123-r*s1*s2*s3<=0
c5: s1*s2*s123-r*s3*s12^2<=0
c6: s1*s3*s123-r*s2*s13^2<=0
c7: s2*s3*s123-r*s1*s23^2<=0
c8: s12^2*s13^2*s23^2
    -r*s1^3*s2^3*s3^3*s123<=0
c9: s1^3*s2*s3*s123-r*s12^2*s13^2<=0
c10: s1*s2^3*s3*s123-r*s12^2*s23^2<=0
c11: s13^2*s23^2-r*s1*s2*s3^3*s123<=0
c12: s1*s2*s3^3*s123-r*s13^2*s23^2<=0
c13: s12^2*s23^2-r*s1*s2^3*s3*s123<=0
c14: s12^2*s13^2-r*s1^3*s2*s3*s123<=0
c15: s1^3*s2^3*s3^3*s123
     -r*s12^2*s13^2*s23^2<=0
c16: s1*s23^2-r*s2*s3*s123<=0
c17: s2*s13^2-r*s1*s3*s123<=0
c18: s3*s12^2-r*s1*s2*s123<=0
c19: s1*s2*s3-r*s123<=0

Bounds
0<s1<2
0<s2<2
0<s3<2
0<s12<3
0<s13<3
0<s23<3
4<=s123<=4
1<=r
End
\end{lstlisting}
\end{minipage}\hfill%
\begin{minipage}[t]{0.475\linewidth}
\begin{lstlisting}[showlines=true, basicstyle=\scriptsize\ttfamily]
SCIP version 6.0.0
Copyright (C) 2002-2018 ZIB Berlin

SCIP> read input.pip

SCIP> opt
SCIP Status  : problem is solved
Solving Time : 2.32
Solving Nodes: 587
Primal Bound : +1.349997004816e+01
Dual Bound   : +1.349997004816e+01
Gap          : 0.00 %

SCIP> display solution
objective value: 13.4999700481621
y                13.4999700481621
s1               0.666630230040994
s2               0.66617295525322
s3               0.667196814705786
s12              0.444091987464541
s13              0.444773633975962
s23              2.1111343785595
s123             4
[...]
\end{lstlisting}
\end{minipage}\vspace{-3mm}
\caption{Computing the minimal absolute interaction for a molecule with 3 sites using {\tt SCIP}.}\label{scipformat}
\end{figure}

Since the number of constraints are not too many when $n=3$, {\tt SCIP} works well. However, for $n=4$ there are $2^{11}$ constraints which makes it a hard problem to solve directly.

To simplify our problem computationally, we compute an upper bound $b_{+}$ and a lower bound $b_{-}$ for the minimal absolute interaction of a given binding polynomial. If both bounds are identical, the minimum is found.
 To compute the upper bound we minimize the minimal absolute interaction in a single region. For the experiments in Subsection~\ref{sectionexperimhemog} we consider the orthant $\mathcal{O}_{\emptyset}$. The result may not be the global minimum - since the minimum may lie in another orthant - but it gives an upper bound.
 For the lower bound, we consider the relaxed problem by looking at a problem with less constraints.  We pick a representative for each of the orbit of the group action $S^{n}$ and consider the corresponding region. We then only consider the constraints given by functions $f_{\mathcal{I}}$ dominant in those regions. Considering less constraints means that the feasible set of the optimization problem is bigger than originally allowed.
 Thus, this new problem gives the lower bound on the true optimization problem.
When the gap between upper and lower bound is sufficiently small, we assume to have found a molecule minimizing the problem.

\section{Experimental Results}\label{sectionexperimhemog}
In this section, we run computational experiments on data from \cite{imai1973analyses,ikeda1983thermodynamic} which is also summarized in \cite{connelly1986analysis}.

\subsection{Data}
The first data set consists of eight oxygen binding polynomials of human adult hemoglobin in its native form and of three chemically treated versions. The oxygen binding to each of the four hemoglobin variants has been observed under two different environmental conditions  \cite{imai1973analyses}.
The coefficients $a_1$, $a_2$ and $a_3$ of the corresponding polynomials $P_1$ to $P_8$ are listed in Figure~\ref{tableconnelly}. The coefficients $a_0$ and $a_4$ are standardized to $1$. Figure~\ref{fig:polynomials1To8} describes the experimental conditions under which binding polynomials $P_1$ to $P_8$ were obtained and illustrates the ranking of the degree of cooperativity as defined by the maximal slope of the Hill plot $n_{max}$ (as reported by \cite{imai1973analyses}). The slope to the Hill plot relates the variance of the probability distribution on the macrostates $\{0,1,2,3,4\}$ at the respective ligand activity to the variance of a binomial distribution with the same mean \cite{abeliovich2005empirical,martini2016cooperative}. A value larger than 1 indicates a variance higher than the maximal value an independent system can generate.\\

\begin{figure}[h]
  \centering
  \begin{tabular}{l|ccc}
    \multirow{2}{*}{Human Hb} & in the absence && in the presence \\
                                      & of 2mM DPG     && of 2mM DPG \\[1mm] \hline
    \multirow{2}{*}{untreated} & \multirow{2}{*}{$P_1$} &\multirow{2}{*}{$\prec$} & \multirow{2}{*}{$P_2$} \\
                                      & & & \\[-2mm]
    & $\curlyvee$ & & $\curlyvee$ \\[-2mm]
    treated with & \multirow{2}{*}{$P_3$} & \multirow{2}{*}{$\prec$} & \multirow{2}{*}{$P_4$} \\
    iodoacetamide & & & \\[-2mm]
    & $\curlyvee$ & & $\curlyvee$ \\[-2mm]
    treated with & \multirow{2}{*}{$P_5$} &\multirow{2}{*}{$\prec$} & \multirow{2}{*}{$P_6$} \\
    N-ethylmaeimide & & & \\[-2mm]
    & $\curlyvee$ & & $\curlyvee$ \\[-2mm]
    treated with & \multirow{2}{*}{$P_7$} &\multirow{2}{*}{$\prec$} & \multirow{2}{*}{$P_8$} \\
    carboxypeptidase A & & &
  \end{tabular}\vspace{-3mm}
  \caption{Conditions under which binding polynomials 1 to 8 were derived and the relation of their degree of cooperativity according to the maximal slope of the Hill plot $n_{max}$ ($\succ$ compares the hill slope $n_{max}$, DPG = 2,3-diphosphoglycerate).}
  \label{fig:polynomials1To8}
\end{figure}

The second data set \cite{ikeda1983thermodynamic} contains five binding polynomials ($P_9$ to $P_{13}$) of native hemoglobin HbII of {\it Scapharca inaequivalvis} measured at different temperatures (Figure~\ref{fig:polynomials9To25}) \cite[Table~3]{ikeda1983thermodynamic}.

\begin{figure}[h]
  \centering
  \begin{tabular}{l|ccccc}
    Clam HbII & at 10${}^\circ$ & at 15${}^\circ$ & at 20${}^\circ$ & at 25${}^\circ$ & at 30${}^\circ$\\[3mm] \hline
                    & $P_{9}$ & $P_{10}$ & $P_{11}$ & $P_{12}$ & $P_{13}$
  \end{tabular}
  \caption{Temperature in degree Celsius at which binding polynomials 9 to 13 were determined.}
  \label{fig:polynomials9To25}
\end{figure}

\subsection{Minimal absolute interactions}

Unlike the case $n=3$ in Example~\ref{example3site}, solving Problem~(\ref{eq:optInS}) for $n=4$ proved to be too hard for \texttt{SCIP}, which is why we list upper ($b_+$) and lower ($b_-$) bounds for the minimal absolute interaction $\|P_i\|$ instead. The upper bound is the minimal absolute interaction inside the region $\mathcal O_{\emptyset}$, i.e., it is obtained by solving Problem \eqref{eq:OptInO} for $\mathcal I=\emptyset$. In other words, the upper bound represents the minimal value on the region in which all interactions are equal to or larger than 1. The lower bound is computed through a relaxation of Problem~(\ref{eq:optInS}). Inspired by the action of the symmetric group $S_n$ on the regions $\mathcal O_{\mathcal I}$, we picked a set of $180$ representatives and discarded all constraints $f_{\mathcal I}(W)\leq r$ with $\mathcal I$ not in them, see Section~\ref{app:representatives}.

\begin{figure}[h]
	\centering
	\begin{tabular}{l|rrr|r|r|r}
		& $a_1$ & $a_2$ & $a_3$  & $b_{+}$ & $b_{-}$ & $n_{max}$\\ \hline
		$P_{1}$  & 0.835 & 0.379 & 0.541  & 527  & 527& 2.51\\
		$P_{2}$  & 0.789 & 0.154 & 0.0648 & 3322 & 1600&3.09 \\
		$P_{3}$  & 1.42  & 2.42  & 0.752  & 111  & 63  &1.63\\
		$P_{4}$  & 0.647 & 0.568 & 0.0986 & 2002& 1460 &2.71\\
		$P_{5}$  & 2.0   & 2.31  & 2.04   & 16  & 16  &1.44\\
		$P_{6}$  & 0.539 & 0.909 & 0.554  & 3033& 3030 &2.27\\
		$P_{7}$  & 3.47  & 4.74  & 2.76   & 2.27    & 1.68   &1.15\\
		$P_{8}$  & 3.26  & 5.36  & 2.23   & 7.64    & 2.19 &1.23\\[3mm]

		$P_{9}$ & 1.4   & 1.0   & 0.62   & 66   & 66  & 2.08\\
		$P_{10}$ & 1.4   & 0.96  & 0.60   & 66   & 66  &2.10\\
		$P_{11}$ & 1.2   & 0.93  & 0.70   & 123  & 123 &2.08\\
		$P_{12}$ & 1.4   & 0.95  & 0.62   & 66   & 66  &2.10\\
		$P_{13}$ & 1.1   & 0.98  & 0.59   & 175  & 175 &2.12
	\end{tabular}
	\caption{Binding polynomials, bounds on the minimal absolute interaction and $n_{max}$ as reported by \cite{ikeda1983thermodynamic,imai1973analyses}. Coefficients $a_0$ and $a_4$ are equal to $1$ for all polynomials.}\label{tableconnelly}
\end{figure}

Figure~\ref{tableconnelly} illustrates that for most binding polynomials our upper and lower bound allows for a sufficient approximation of the absolute minimal interaction.
However, there remains a nontrivial gap between both bounds for $P_2$, $P_3$, $P_4$.
Let us assume that the correct value for $P_2$ is close to the upper bound which we determined. Then, nearly all relations of the degree of cooperativity of polynomials $P_1$ to $P_8$ described by the maximal slope of the Hill plot $n_{max}$ (Figure~\ref{fig:polynomials1To8}) are also found with the minimal interaction criterion. The only exception is the ranking of $P_4$ and $P_6$.
Whereas $P_4$ describes a more cooperative system than $P_6$ according to $n_{max}$, we obtain a different picture when using the minimal interaction criterion which states that $P_6$ is almost as cooperative as $P_2$. The question arises why the two measures of cooperativity are so different.
An answer may be found considering the coefficients of $P_4$ and $P_6$. The coefficient $a_2$ of $P_6$ is larger than $a_1$ and $a_3$, which indicates that the macrostate of having $2$ sites occupied --at sufficiently high ligand activity-- is more probable than the states of having $1$ or $3$ ligands bound. To generate this type of distribution, we may need absolute interaction energy at each macrosotate level, for instance negative and positive interaction energy. Moreover, a family of distributions with a relatively high weight on macrostate $2$ has a reduced possibility to generate a high variance since an extraordinary high variance requires more weights on the macrostates $0$ and $4$.\\

The polynomials $P_9$ to $P_{13}$ were derived from the same protein at different temperatures \cite{ikeda1983thermodynamic}. We see that the reported $n_{max}$ is not behaving monotonously with increasing temperature, but that it varies between $2.08$ and $2.12$. Since one might expect a more monotonous development, we interpret this observation as a result of measuring imprecision of the coefficients. Comparing the minimal interaction to $n_{max}$, we see the same problem on a larger scale, which is a results of the minimal energy criterion being defined on an exponential scale of actual physical energies.

We consider the molecules realizing the upper bound in more detail.

\subsection{Molecules solving the optimization problem}

Figure~\ref{fig:minimalMolecules4Sites} presents the molecules realizing the upper bound of the minimal absolute interaction for all polynomials.
We see that the major part of the estimated interactions is very close to $1$. This is not surprising since the considered optimization problem minimizes the absolute values of the physical energies required to generate the binding curve. $L_1$ optimizations --for instance such as the least absolute shrinkage and selection operator (Lasso) regression \cite{tibshirani1996regression}-- are known to provide sparse solutions, meaning that many of the physical energies of the solution will be zero. Thus, this observation fits well to the context. Considering the structure of the solutions in more detail, we see that most solutions include a non-trivial weight for a second degree interaction ($w_{34}$).
The only polynomial whose minimizing molecule in $\mathcal{O}_{\emptyset}$ does not possess an interaction of pair-wise interaction is $P_2$
which has the highest $n_{max}$. This observation seems plausible since the weights should be more towards the extreme values, that is to the macrostates $0$ and $4$ to obtain a high variance of the macroscopic distribution. Compared to $P_1$, $P_2$ puts more probability on the macrostates $3$ and $4$. With this view, we can also explain the apparent contradiction that we find when using $n_{max}$ or the minimal absolute interaction to compare polynomials $P_4$ and $P_6$. As described above, $P_6$ has a coefficient $a_2$ which is larger than $a_1$ and $a_3$ and almost of same size as $a_4$, which is $1$.
This means that the distribution on the macrostates is stabilized towards the macrostate $2$, which hinders
very large variances and thus leads to a smaller $n_{max}$ than that of $P_4$.
However, this circumstance leads to a high interaction weight for a microstate of macrostate $2$. While the main weight for $P_4$ is on macrostate $4$ -which helps to achieve a higher value of $n_{max}$, $P_6$ requires more interaction energy to model the behavior at macrostate $2$. \\

\begin{figure}[h]
	\centering
  \setlength{\tabcolsep}{2.5pt}
	\begin{tabular}{l|c c c c c c c c c c c}
		& $w_1$ & $w_2$ & $w_3$  & $w_{4}$ & $w_{12}\cdots w_{24}$ & $w_{34}$ & $w_{123}$ & $w_{124}$ & $w_{134}$ & $w_{234}$ & $w_{1234}$ \\ \hline

		$P_{1}$  & 0.209 & 0.209 & 0.209 & 0.209 & 1.00 & \color{blue} 3.70 & 1.00 & 1.00 & \color{blue}14.5 & 1.00 & \color{blue} 9.79 \\

		$P_{2}$ & 0.0822 & 0.0822 & 0.0822 &  0.543 & 1.00 & 1.00 & \color{blue} 6.60 & 1.00 & \color{blue} 14.7 & 1.00 & \color{blue} 34.3 \\

		$P_{3}$ & 0.179 & 0.179 & 0.531 & 0.531 & 1.00 & \color{blue} 7.12 & 1.00 & 1.00 & 1.00 & 1.00 &\color{blue} 15.6 \\

		$P_{4}$ & 0.101 & 0.101 & 0.223 & 0.223 & 1.00 & \color{blue} 9.42 &  1.00 & 1.00 & 1.00 & 1.00 & \color{blue}212 \\

		$P_{5}$ & 0.500 & 0.500 & 0.500 & 0.500 & 1.00 & \color{blue} 4.24 & 1.00 & 1.00 & \color{blue} 1.69 & \color{blue} 1.69 & \color{blue} 1.32 \\

		$P_{6}$ & 0.135 & 0.135 & 0.135 & 0.135 & 1.00 & \color{blue} 45.1 &  1.00 & 1.00 & \color{blue} 2.49 & \color{blue} 2.49 &\color{blue} 10.9 \\

		$P_{7}$ & 0.569 & 0.569 & 1.17 & 1.17 & 1.00 & \color{blue} 1.30 &1.00 & 1.00 & 1.00 & 1.00 &\color{blue} 1.75 \\

		$P_{8}$ & 0.265 & 0.265 & 1.36 & 1.36 & 1.00 & \color{blue} 2.06 & 1.00 & 1.00 & 1.00 & 1.00 & \color{blue} 3.70 \\[3mm]

		$P_{9}$ & 0.350 & 0.350 & 0.350 & 0.350 & 1.00 & \color{blue} 3.16 & 1.00 & 1.00 &\color{blue} 1.97 & \color{blue} 1.97 & \color{blue} 5.43 \\

		$P_{10}$ & 0.350 & 0.350 & 0.350 & 0.350 & 1.00 & \color{blue} 2.84 &1.00 & 1.00 &\color{blue} 2.11 &\color{blue} 2.11 &\color{blue} 5.26\\

		$P_{11}$ & 0.300 & 0.300 & 0.300 & 0.300 & 1.00 & \color{blue} 5.33 & 1.00 & 1.00 & \color{blue} 2.24 &\color{blue} 2.24 & \color{blue} 4.60 \\

		$P_{12}$ & 0.350 & 0.350 & 0.350 & 0.300 & 1.00 & \color{blue} 2.76 & 1.00 & 1.00 &\color{blue} 2.26 &\color{blue} 2.26 & \color{blue} 4.73 \\

		$P_{13}$ & 0.275 & 0.275 & 0.275 & 0.275 & 1.00 &\color{blue} 7.96 & 1.00 & 1.00 & \color{blue} 1.66 & \color{blue} 1.66 &\color{blue} 8.01
	\end{tabular}
	\caption{Molecules of polynomials $P_1$ to $P_{13}$ realizing the upper bound $b_+$. The $w_I$s are rounded to four digits. The incidences of a $w_I$ being smaller than $1$ are numerical impressions when resolving from the $s_I$ coordinates. All interactions larger than 1 are highlighted in blue.} \label{fig:minimalMolecules4Sites}
\end{figure}

\section{Discussion and Outlook}
\subsection{Cooperativity and minimal absolute interaction}
A commonly used macroscopic conceptualization of cooperativity is that ``It's all or nothing'' \cite{hunter2009cooperativity}, meaning for hemoglobin that the probability weights of the extreme macrostates $0$ and $4$ are relatively big. This definition also aligns well with using the maximal Hill slope as criterion of cooperativity, since it detects a variance which is higher than that of any independent system.
However, observations made on the macroscopic level have often been interpreted with a view on microscopic properties, for instance that ``binding of a ligand molecule increases the receptor's apparent affinity, and hence increases the chance of another ligand molecule binding'' \cite{stefan2013cooperative}. Connecting macroscopic and microscopic properties is a non-trivial problem if no additional restriction on the underlying molecule are made. Microscopic properties can be easily inferred from the macroscopic behavior if additional assumptions are made, for instance the molecule being composed of physically indistinguishable sites. In particular, given a binding polynomial, we can calculate all interaction energies easily when we assume that $w_{ij} = w_{kl}$, $w_{ijk} = w_{lmn}$, etc.
However, without additional knowledge or restrictions, it is difficult to relate macroscopic and microscopic properties of larger, asymmetric systems \cite{martini2016cooperative}.
The minimal interaction energy required to generate the observed binding curve is a good candidate to connect macroscopic and microscopic characteristics and additional side conditions based on knowledge of the molecule structure could also be incorporated \cite{martini2017measure}. \\

As illustrated by the examples of $P_4$ and $P_6$, the minimal interaction energy criterion captures a conceptually different type of cooperativity than $n_{max}$. The polynomial $P_6$ requires a relatively high interaction energy due to giving the highest weights to macrostates $0$, $2$ and $4$, but stabilizing towards the intermediate value $2$ prevents an extreme variance. Contrary, $P_4$ has monotonously decreasing coefficients $a_1$ to $a_3$. Here, modeling the distribution requires less interaction energy to capture the relation of the probability on macrostates $1$, $2$ and $3$, and a higher value of the variance of the distribution can be realized.

\subsection{The mathematical problem}
Ligand binding raises many intriguing questions in applied algebra and geometry, and stands to benefit equally from the provided insight into the intrinsic geometry of the problem.
While previous works have explored the application of polynomial system solving techniques to a situation with different types of ligands \cite{Ren2019,martini2013interaction}, this work explores the use of polynomial optimization.
Using the minimal interaction energy as a criterion for cooperativity without assuming the target molecule to be composed of physically identical units, we deal with an algebraic optimization problem.
Even though there exists dedicated software to solve this kind of problem, we had to apply a transformation of the coordiantes (from $w_I$ to $s_I$), a transformation of the optimization problem (Eq.(\ref{newoptprob})) and approximations by upper and lower bounds to approach the considered examples. This effort illustrates that the setup of looking for the minimal absolute interaction by this method becomes rapidly complicated with increasing number of binding sites.
Whereas the case of three binding sites can be solved easily with \texttt{SCIP}, tetrameric hemoglobin with four binding sites already posed computational problems.
A way to facilitate the problem may be to exploit its intrinsic structure. For instance, there are results for strongly symmetric systems which might be generalized and tailored to the system with which we are dealing  \cite{RienerSafey18}.
In addition, other mathematical fields such as tropical geometry \cite{maclagan2015introduction} may provide interesting tools to approach this problem.

\subsection{Other mathematical perspectives}
\subsubsection{Tensors}
Tensors \cite{Landsberg12} are multidimensional arrays of real numbers which generalize matrices. We can regard molecules as (positive) tensors in $\RR^{2\times \dots\times 2} = (\RR^2)^{\otimes n}$, whose entries are indexed by $\{0,1\}^n$ which is in bijection with the subsets of $[n]$.

From this point of view, many concepts in tensor theory have direct analogues in ligand binding, such as symmetric tensors and molecules with indistinguishable sites. Notably, the central concepts of tensor rank and tensor rank decomposition have a special meaning in ligand binding:

\begin{proposition}
  Let $S=(s_I)_{I\subseteq[n]}\in\RR^{2^n}$ be a molecule with $n$ sites and whose corresponding tensor $(s_I)_{I\in\{0,1\}^n}$ has (positive) rank $k$. Then $S$ can be written as a weighted average of $k$ molecules with $n$ independent sites. In particular, a corresponding mixture of the $k$ molecules will have the same binding curve as $S$.
  Moreover, this is not possible with $k-1$ molecules with $n$ independent sites.
\end{proposition}
\begin{proof}
  Rank-one tensors are of the form $v_1\otimes \dots \otimes v_n$ for some $v_i=(s_i,t_i)\in\RR^2$ with entries
  \begin{equation*}
    \label{eq:tensor}
    (1,0,\dots)\!: s_1\cdot\prod_{i\neq 1} t_i, \quad (0,1,0,\dots)\!: s_2 \cdot\prod_{i\neq 2} t_i, \quad (1,1,0,\dots)\!: s_1s_2 \cdot\prod_{i\neq 1,2} t_i, \quad \ldots
  \end{equation*}
  They can be regarded as molecules with particularly simple interaction, as molecules with independent sites are rank-one tensors with $t_i=1$. A tensor $S$ of rank $k$ can be expressed as
  \[ S = \sum_{j=1}^k v_1^{(j)} \otimes \dots \otimes v_n^{(j)} \]
  for $k$ rank-one tensors $v_1^{(j)} \otimes \dots \otimes v_n^{(j)}$ and cannot be expressed as a sum of $k-1$ rank-one tensors. Supposing $v_i^{(j)}=(s_i^{(j)},t_i^{(j)})\in\RR^2$, we can reformulate this into
  \[ S = \sum_{j=1}^k \Big(\prod_{i=1}^k t_i^{(j)} \Big)
    \cdot \left(
      \begin{smallmatrix}
        s_1^{(j)}/t_1^{(j)} \\ 1
      \end{smallmatrix}
    \right) \otimes \dots \otimes \left(
      \begin{smallmatrix}
        s_n^{(j)}/t_n^{(j)} \\ 1
      \end{smallmatrix}
    \right), \]
  thus $S$ is a weighted average $k$ molecules with $n$ independent sites.
\end{proof}



\subsubsection{Tropical Geometry}
Tropical geometry \cite{maclagan2015introduction} revolves around functions on the max-plus semiring. One example of such as function is the absolute interaction in logarithmic coordinates:
\begin{center}
  \begin{tikzpicture}
    \node (topLeft) {$(\RR_{>0})^{2^n}$};
    \node[anchor=base west,xshift=45mm] (topRight) at (topLeft.base east) {$\RR^{2^n}$};
    \node[anchor=base,yshift=-15mm] (bottomRight) at (topRight.base) {$(\log(w_I))$};
    \node[anchor=base west,xshift=-2mm] at (bottomRight.base east) {$=:(x_I)$};
    \node[anchor=base,yshift=-15mm] (bottomLeft) at (topLeft.base) {$(w_I)$};
    \node[anchor=base,yshift=-15mm] at (bottomLeft.base) {$\displaystyle\prod_{I\subseteq[n]}\max(w_I,w_I^{-1})$};
    \node[anchor=base,yshift=-15mm] (bottombottomRight) at (bottomRight.base) {$\displaystyle\sum_{I\subseteq[n]}\max(x_I,-x_I)$};
    \node[anchor=base west,xshift=-2mm] at (bottombottomRight.base east) {$=:\|(x_I)\|$};

    \draw[<->] ($(topLeft.base east)+(1.25,0.1)$) -- ($(topRight.base west)+(-1.25,0.1)$);
    \draw[|->] ($(topLeft.base east)+(1.25,-1.4)$) -- ($(topRight.base west)+(-1.25,-1.4)$);
    \draw[|->] ($(topLeft.base east)+(1.25,-2.9)$) -- ($(topRight.base west)+(-1.25,-2.9)$);

    \draw[draw opacity=0]
    (bottomLeft) -- node[sloped] {$\in$} (topLeft)
    (bottomRight) -- node[sloped] {$\in$} (topRight);
  \end{tikzpicture}
\end{center}

In tropical geometry, the function $\|\cdot\|$ is referred to as a tropical rational function. It is a piecewise-linear function, linear on the image of each $\mathcal O_{\mathcal I}$, and the boundary separating its regions of linearity is called its tropical hypersurface $\text{Trop}(\|\cdot\|)$. The gradient of $\|\cdot\|$ is constant on each region of linearity, which is why understanding the intersection of a set $L\subseteq\RR^{2^n}$ with $\text{Trop}(\|\cdot\|)$ is paramount to computing the minimum of $\|\cdot\|$ on $L$.

In our setting, $L$ is the image of the feasible set given by the conditions $\Phi(W)=P$ which is an exponential variety of codimension $n$. Since the number of regions of linearity is finite, there exist only a finite number of intersection patterns of $L$ and $\text{Trop}(\|\cdot\|)$. Classifying these intersection patterns depending on the coefficients of $P$ will contribute greatly towards computing the minimal interactions and minimizing molecules. 

\subsubsection{Dependency measures}
From a stochastic or statistical point of view, the problem of cooperativity is a problem of quantifying (minimal) dependency, which is a classical research topic in these fields \cite{renyi1959measures,schweizer1981nonparametric,koyak1987measuring}. In particular copulas \cite{schweizer1991thirty,durante2010copula,nelsen2007introduction} have become a common tool to describe dependencies between random variables. Copulas cannot be applied directly to the problem since they are defined on continuous, not discrete, random variables. Moreover, they aim at modeling the dependency of given random variables which in our scenario corresponds to the dependency of the sites of a given molecule. The minimization problem of finding the minimal dependency required to generate a given distribution of a sum of dependent variables is usually not included. Nevertheless, the problem of cooperativity should also be treated from a more probabilistic point of view and the available concepts should be considered in more detail in the context of ligand binding.

\subsection{Conclusion}
We illustrated how searching for the minimal interaction energy required to generate a binding curve leads to a problem of algebraic optimization. Even though, there are computational tools available and the structure of the problem seems relatively simple on first sight, it becomes complicated when more than $3$ binding sites are considered. In particular the fact that we had to calculate a lower bound by a relaxed problem and not having been able to close the gap between upper and lower bound completely, indicates that our approach can currently not be enveloped in a simple ready-to-use tool to quantify cooperativity. We showed that the concept of minimally present interaction energy is different from using the maximal Hill slope as a measure of cooperativity. The first also detects required changes in affinity from one macrostate to the other, whereas the latter only focuses on the macroscopic variance of the distribution. Future work may use other mathematical concepts such as tensor theory, tropical geometry and probabilistic measures of dependencies to approach questions related to cooperativity.


\section{List of Constraints}\label{app:representatives}

Below is a compact representation of the list of $180$ constraints we considered to compute the lower bounds $b_{-}$ of minimal absolute interaction in Figure~\ref{tableconnelly}.
There, $\{12,34,123\}$ represents the constraint $f_{\mathcal I}\leq r$ for $\mathcal I=\{12,34,123\}$, which is the absolute interaction $\|W\|$ for $W\in\mathcal O_{\mathcal I}$ where $w_{12},w_{34},w_{123}\leq 1$ and $w_{I} \geq 1$ otherwise. 

\vspace{0.7cm}

\noindent
\begin{scriptsize}
$\emptyset$, $\!\{12\}$, $\!\{123\}$, $\!\{1234\}$, $\!\{12$, $\!\!13\}$, $\!\{12$, $\!\!34\}$, $\!\{12$, $\!\!123\}$, $\!\{12$, $\!\!134\}$, $\!\{123$, $\!\!124\}$, $\!\{12$, $\!\!1234\}$, $\!\{123$, $\!\!1234\}$, $\!\{12$, $\!\!13$, $\!\!14\}$, $\!\{12$, $\!\!13$, $\!\!23\}$, $\!\{12$, $\!\!13$, $\!\!24\}$, $\!\{12$, $\!\!13$, $\!\!123\}$, $\!\{12$, $\!\!13$, $\!\!124\}$, $\!\!\!\{12$, $\!\!34$, $\!\!123\}$, $\!\{12$, $\!\!13$, $\!\!234\}$, $\!\{12$, $\!\!123$, $\!\!124\}$, $\!\{12$, $\!\!123$, $\!\!134\}$, $\!\{12$, $\!\!134$, $\!\!234\}$, $\!\{123$, $\!\!124$, $\!\!134\}$, $\!\{12$, $\!\!13$, $\!\!1234\}$, $\!\{12$, $\!\!34$, $\!\!1234\}$, $\!\{12$, $\!\!123$, $\!\!1234\}$, $\!\{12$, $\!\!134$, $\!\!1234\}$, $\!\{123$, $\!\!124$, $\!\!1234\}$, $\!\{12$, $\!\!13$, $\!\!14$, $\!\!23\}$, $\!\{12$, $\!\!13$, $\!\!24$, $\!\!34\}$, $\!\{12$, $\!\!13$, $\!\!14$, $\!\!123\}$, $\!\{12$, $\!\!13$, $\!\!23$, $\!\!123\}$, $\!\{12$, $\!\!13$, $\!\!24$, $\!\!123\}$, $\!\{12$, $\!\!13$, $\!\!23$, $\!\!124\}$, $\!\{12$, $\!\!13$, $\!\!24$, $\!\!134\}$, $\!\{12$, $\!\!13$, $\!\!14$, $\!\!234\}$, $\!\{12$, $\!\!13$, $\!\!123$, $\!\!124\}$, $\!\{12$, $\!\!13$, $\!\!124$, $\!\!134\}$, $\!\{12$, $\!\!13$, $\!\!123$, $\!\!234\}$, $\!\{12$, $\!\!34$, $\!\!123$, $\!\!134\}$, $\!\{12$, $\!\!34$, $\!\!123$, $\!\!124\}$, $\!\{12$, $\!\!13$, $\!\!124$, $\!\!234\}$, $\!\{12$, $\!\!123$, $\!\!124$, $\!\!134\}$, $\!\{12$, $\!\!123$, $\!\!134$, $\!\!234\}$, $\!\{123$, $\!\!124$, $\!\!134$, $\!\!234\}$, $\!\{12$, $\!\!13$, $\!\!14$, $\!\!1234\}$, $\!\{12$, $\!\!13$, $\!\!23$, $\!\!1234\}$, $\!\{12$, $\!\!13$, $\!\!24$, $\!\!1234\}$, $\!\{12$, $\!\!13$, $\!\!123$, $\!\!1234\}$, $\!\{12$, $\!\!13$, $\!\!124$, $\!\!1234\}$, $\!\{12$, $\!\!34$, $\!\!123$, $\!\!1234\}$, $\!\{12$, $\!\!13$, $\!\!234$, $\!\!1234\}$, $\!\{12$, $\!\!123$, $\!\!124$, $\!\!1234\}$, $\!\!\!\{12$, $\!\!123$, $\!\!134$, $\!\!1234\}$, $\!\{12$, $\!\!134$, $\!\!234$, $\!\!1234\}$, $\!\{123$, $\!\!124$, $\!\!134$, $\!\!1234\}$, $\!\{12$, $\!\!13$, $\!\!14$, $\!\!23$, $\!\!24\}$, $\!\{12$, $\!\!13$, $\!\!14$, $\!\!23$, $\!\!123\}$, $\!\{12$, $\!\!13$, $\!\!14$, $\!\!23$, $\!\!124\}$, $\!\{12$, $\!\!13$, $\!\!24$, $\!\!34$, $\!\!123\}$, $\!\{12$, $\!\!13$, $\!\!14$, $\!\!23$, $\!\!234\}$, $\!\{12$, $\!\!13$, $\!\!14$, $\!\!123$, $\!\!124\}$, $\!\{12$, $\!\!13$, $\!\!23$, $\!\!123$, $\!\!124\}$, $\!\{12$, $\!\!13$, $\!\!24$, $\!\!123$, $\!\!124\}$, $\!\{12$, $\!\!13$, $\!\!24$, $\!\!123$, $\!\!234\}$, $\!\{12$, $\!\!13$, $\!\!24$, $\!\!123$, $\!\!134\}$, $\!\{12$, $\!\!13$, $\!\!23$, $\!\!124$, $\!\!134\}$, $\!\{12$, $\!\!13$, $\!\!14$, $\!\!123$, $\!\!234\}$, $\!\{12$, $\!\!13$, $\!\!24$, $\!\!134$, $\!\!234\}$, $\!\{12$, $\!\!13$, $\!\!123$, $\!\!124$, $\!\!134\}$, $\!\{12$, $\!\!13$, $\!\!123$, $\!\!124$, $\!\!234\}$, $\!\{12$, $\!\!34$, $\!\!123$, $\!\!124$, $\!\!134\}$, $\!\{12$, $\!\!13$, $\!\!124$, $\!\!134$, $\!\!234\}$, $\!\{12$, $\!\!123$, $\!\!124$, $\!\!134$, $\!\!234\}$, $\!\{12$, $\!\!13$, $\!\!14$, $\!\!23$, $\!\!1234\}$, $\!\{12$, $\!\!13$, $\!\!24$, $\!\!34$, $\!\!1234\}$, $\!\{12$, $\!\!13$, $\!\!14$, $\!\!123$, $\!\!1234\}$, $\!\{12$, $\!\!13$, $\!\!23$, $\!\!123$, $\!\!1234\}$, $\!\{12$, $\!\!13$, $\!\!24$, $\!\!123$, $\!\!1234\}$, $\!\{12$, $\!\!13$, $\!\!23$, $\!\!124$, $\!\!1234\}$, $\!\{12$, $\!\!13$, $\!\!24$, $\!\!134$, $\!\!1234\}$, $\!\{12$, $\!\!13$, $\!\!14$, $\!\!234$, $\!\!1234\}$, $\!\{12$, $\!\!13$, $\!\!123$, $\!\!124$, $\!\!1234\}$, $\!\{12$, $\!\!13$, $\!\!124$, $\!\!134$, $\!\!1234\}$, $\!\{12$, $\!\!13$, $\!\!123$, $\!\!234$, $\!\!1234\}$, $\!\{12$, $\!\!34$, $\!\!123$, $\!\!134$, $\!\!1234\}$, $\!\{12$, $\!\!34$, $\!\!123$, $\!\!124$, $\!\!1234\}$, $\!\{12$, $\!\!13$, $\!\!124$, $\!\!234$, $\!\!1234\}$, $\!\{12$, $\!\!123$, $\!\!124$, $\!\!134$, $\!\!1234\}$, $\!\{12$, $\!\!123$, $\!\!134$, $\!\!234$, $\!\!1234\}$, $\!\{123$, $\!\!124$, $\!\!134$, $\!\!234$, $\!\!1234\}$, $\!\{12$, $\!\!13$, $\!\!14$, $\!\!23$, $\!\!24$, $\!\!34\}$, $\!\{12$, $\!\!13$, $\!\!14$, $\!\!23$, $\!\!24$, $\!\!123\}$, $\!\{12$, $\!\!13$, $\!\!14$, $\!\!23$, $\!\!24$, $\!\!134\}$, $\!\{12$, $\!\!13$, $\!\!14$, $\!\!23$, $\!\!123$, $\!\!124\}$, $\!\{12$, $\!\!13$, $\!\!24$, $\!\!34$, $\!\!123$, $\!\!234\}$, $\!\{12$, $\!\!13$, $\!\!14$, $\!\!23$, $\!\!124$, $\!\!134\}$, $\!\{12$, $\!\!13$, $\!\!14$, $\!\!23$, $\!\!123$, $\!\!234\}$, $\!\{12$, $\!\!13$, $\!\!24$, $\!\!34$, $\!\!123$, $\!\!124\}$, $\!\{12$, $\!\!13$, $\!\!14$, $\!\!23$, $\!\!124$, $\!\!234\}$, $\!\{12$, $\!\!13$, $\!\!14$, $\!\!123$, $\!\!124$, $\!\!134\}$, $\!\{12$, $\!\!13$, $\!\!23$, $\!\!123$, $\!\!124$, $\!\!134\}$, $\!\{12$, $\!\!13$, $\!\!24$, $\!\!123$, $\!\!124$, $\!\!134\}$, $\!\{12$, $\!\!13$, $\!\!14$, $\!\!123$, $\!\!124$, $\!\!234\}$, $\!\{12$, $\!\!13$, $\!\!24$, $\!\!123$, $\!\!134$, $\!\!234\}$, $\!\{12$, $\!\!13$, $\!\!23$, $\!\!124$, $\!\!134$, $\!\!234\}$, $\!\{12$, $\!\!13$, $\!\!123$, $\!\!124$, $\!\!134$, $\!\!234\}$, $\!\{12$, $\!\!34$, $\!\!123$, $\!\!124$, $\!\!134$, $\!\!234\}$, $\!\{12$, $\!\!13$, $\!\!14$, $\!\!23$, $\!\!24$, $\!\!1234\}$, $\!\{12$, $\!\!13$, $\!\!14$, $\!\!23$, $\!\!123$, $\!\!1234\}$, $\!\{12$, $\!\!13$, $\!\!14$, $\!\!23$, $\!\!124$, $\!\!1234\}$, $\!\{12$, $\!\!13$, $\!\!24$, $\!\!34$, $\!\!123$, $\!\!1234\}$, $\!\{12$, $\!\!13$, $\!\!14$, $\!\!23$, $\!\!234$, $\!\!1234\}$, $\!\{12$, $\!\!13$, $\!\!14$, $\!\!123$, $\!\!124$, $\!\!1234\}$, $\!\{12$, $\!\!13$, $\!\!23$, $\!\!123$, $\!\!124$, $\!\!1234\}$, $\!\{12$, $\!\!13$, $\!\!24$, $\!\!123$, $\!\!124$, $\!\!1234\}$, $\!\{12$, $\!\!13$, $\!\!24$, $\!\!123$, $\!\!234$, $\!\!1234\}$, $\!\{12$, $\!\!13$, $\!\!24$, $\!\!123$, $\!\!134$, $\!\!1234\}$, $\!\{12$, $\!\!13$, $\!\!23$, $\!\!124$, $\!\!134$, $\!\!1234\}$, $\!\{12$, $\!\!13$, $\!\!14$, $\!\!123$, $\!\!234$, $\!\!1234\}$, $\!\{12$, $\!\!13$, $\!\!24$, $\!\!134$, $\!\!234$, $\!\!1234\}$, $\!\{12$, $\!\!13$, $\!\!123$, $\!\!124$, $\!\!134$, $\!\!1234\}$, $\!\{12$, $\!\!13$, $\!\!123$, $\!\!124$, $\!\!234$, $\!\!1234\}$, $\!\{12$, $\!\!34$, $\!\!123$, $\!\!124$, $\!\!134$, $\!\!1234\}$, $\!\{12$, $\!\!13$, $\!\!124$, $\!\!134$, $\!\!234$, $\!\!1234\}$, $\!\{12$, $\!\!123$, $\!\!124$, $\!\!134$, $\!\!234$, $\!\!1234\}$, $\!\{12$, $\!\!13$, $\!\!14$, $\!\!23$, $\!\!24$, $\!\!34$, $\!\!123\}$, $\!\{12$, $\!\!13$, $\!\!14$, $\!\!23$, $\!\!24$, $\!\!123$, $\!\!124\}$, $\!\{12$, $\!\!13$, $\!\!14$, $\!\!23$, $\!\!24$, $\!\!123$, $\!\!134\}$, $\!\{12$, $\!\!13$, $\!\!14$, $\!\!23$, $\!\!24$, $\!\!134$, $\!\!234\}$, $\!\{12$, $\!\!13$, $\!\!14$, $\!\!23$, $\!\!123$, $\!\!124$, $\!\!134\}$, $\!\{12$, $\!\!13$, $\!\!14$, $\!\!23$, $\!\!123$, $\!\!124$, $\!\!234\}$, $\!\{12$, $\!\!13$, $\!\!24$, $\!\!34$, $\!\!123$, $\!\!124$, $\!\!134\}$, $\!\{12$, $\!\!13$, $\!\!14$, $\!\!23$, $\!\!124$, $\!\!134$, $\!\!234\}$, $\!\{12$, $\!\!13$, $\!\!14$, $\!\!123$, $\!\!124$, $\!\!134$, $\!\!234\}$, $\!\{12$, $\!\!13$, $\!\!23$, $\!\!123$, $\!\!124$, $\!\!134$, $\!\!234\}$, $\!\{12$, $\!\!13$, $\!\!24$, $\!\!123$, $\!\!124$, $\!\!134$, $\!\!234\}$, $\!\{12$, $\!\!13$, $\!\!14$, $\!\!23$, $\!\!24$, $\!\!34$, $\!\!1234\}$, $\!\{12$, $\!\!13$, $\!\!14$, $\!\!23$, $\!\!24$, $\!\!123$, $\!\!1234\}$, $\!\{12$, $\!\!13$, $\!\!14$, $\!\!23$, $\!\!24$, $\!\!134$, $\!\!1234\}$, $\!\{12$, $\!\!13$, $\!\!14$, $\!\!23$, $\!\!123$, $\!\!124$, $\!\!1234\}$, $\!\{12$, $\!\!13$, $\!\!24$, $\!\!34$, $\!\!123$, $\!\!234$, $\!\!1234\}$, $\!\{12$, $\!\!13$, $\!\!14$, $\!\!23$, $\!\!124$, $\!\!134$, $\!\!1234\}$, $\!\{12$, $\!\!13$, $\!\!14$, $\!\!23$, $\!\!123$, $\!\!234$, $\!\!1234\}$, $\!\{12$, $\!\!13$, $\!\!24$, $\!\!34$, $\!\!123$, $\!\!124$, $\!\!1234\}$, $\!\{12$, $\!\!13$, $\!\!14$, $\!\!23$, $\!\!124$, $\!\!234$, $\!\!1234\}$, $\!\{12$, $\!\!13$, $\!\!14$, $\!\!123$, $\!\!124$, $\!\!134$, $\!\!1234\}$, $\!\{12$, $\!\!13$, $\!\!23$, $\!\!123$, $\!\!124$, $\!\!134$, $\!\!1234\}$, $\!\{12$, $\!\!13$, $\!\!24$, $\!\!123$, $\!\!124$, $\!\!134$, $\!\!1234\}$, $\!\{12$, $\!\!13$, $\!\!14$, $\!\!123$, $\!\!124$, $\!\!234$, $\!\!1234\}$, $\!\{12$, $\!\!13$, $\!\!24$, $\!\!123$, $\!\!134$, $\!\!234$, $\!\!1234\}$, $\!\{12$, $\!\!13$, $\!\!23$, $\!\!124$, $\!\!134$, $\!\!234$, $\!\!1234\}$, $\!\{12$, $\!\!13$, $\!\!123$, $\!\!124$, $\!\!134$, $\!\!234$, $\!\!1234\}$, $\!\{12$, $\!\!34$, $\!\!123$, $\!\!124$, $\!\!134$, $\!\!234$, $\!\!1234\}$, $\!\{12$, $\!\!13$, $\!\!14$, $\!\!23$, $\!\!24$, $\!\!34$, $\!\!123$, $\!\!124\}$, $\!\{12$, $\!\!13$, $\!\!14$, $\!\!23$, $\!\!24$, $\!\!123$, $\!\!124$, $\!\!134\}$, $\!\{12$, $\!\!13$, $\!\!14$, $\!\!23$, $\!\!24$, $\!\!123$, $\!\!134$, $\!\!234\}$, $\!\{12$, $\!\!13$, $\!\!14$, $\!\!23$, $\!\!123$, $\!\!124$, $\!\!134$, $\!\!234\}$, $\!\{12$, $\!\!13$, $\!\!24$, $\!\!34$, $\!\!123$, $\!\!124$, $\!\!134$, $\!\!234\}$, $\!\{12$, $\!\!13$, $\!\!14$, $\!\!23$, $\!\!24$, $\!\!34$, $\!\!123$, $\!\!1234\}$, $\!\{12$, $\!\!13$, $\!\!14$, $\!\!23$, $\!\!24$, $\!\!123$, $\!\!124$, $\!\!1234\}$, $\!\{12$, $\!\!13$, $\!\!14$, $\!\!23$, $\!\!24$, $\!\!123$, $\!\!134$, $\!\!1234\}$, $\!\{12$, $\!\!13$, $\!\!14$, $\!\!23$, $\!\!24$, $\!\!134$, $\!\!234$, $\!\!1234\}$, $\!\{12$, $\!\!13$, $\!\!14$, $\!\!23$, $\!\!123$, $\!\!124$, $\!\!134$, $\!\!1234\}$, $\!\{12$, $\!\!13$, $\!\!14$, $\!\!23$, $\!\!123$, $\!\!124$, $\!\!234$, $\!\!1234\}$, $\!\{12$, $\!\!13$, $\!\!24$, $\!\!34$, $\!\!123$, $\!\!124$, $\!\!134$, $\!\!1234\}$, $\!\{12$, $\!\!13$, $\!\!14$, $\!\!23$, $\!\!124$, $\!\!134$, $\!\!234$, $\!\!1234\}$, $\!\{12$, $\!\!13$, $\!\!14$, $\!\!123$, $\!\!124$, $\!\!134$, $\!\!234$, $\!\!1234\}$, $\!\{12$, $\!\!13$, $\!\!23$, $\!\!123$, $\!\!124$, $\!\!134$, $\!\!234$, $\!\!1234\}$, $\!\{12$, $\!\!13$, $\!\!24$, $\!\!123$, $\!\!124$, $\!\!134$, $\!\!234$, $\!\!1234\}$, $\!\{12$, $\!\!13$, $\!\!14$, $\!\!23$, $\!\!24$, $\!\!34$, $\!\!123$, $\!\!124$, $\!\!134\}$, $\!\{12$, $\!\!13$, $\!\!14$, $\!\!23$, $\!\!24$, $\!\!123$, $\!\!124$, $\!\!134$, $\!\!234\}$, $\!\{12$, $\!\!13$, $\!\!14$, $\!\!23$, $\!\!24$, $\!\!34$, $\!\!123$, $\!\!124$, $\!\!1234\}$, $\!\{12$, $\!\!13$, $\!\!14$, $\!\!23$, $\!\!24$, $\!\!123$, $\!\!124$, $\!\!134$, $\!\!1234\}$, $\!\{12$, $\!\!13$, $\!\!14$, $\!\!23$, $\!\!24$, $\!\!123$, $\!\!134$, $\!\!234$, $\!\!1234\}$, $\!\{12$, $\!\!13$, $\!\!14$, $\!\!23$, $\!\!123$, $\!\!124$, $\!\!134$, $\!\!234$, $\!\!1234\}$, $\!\{12$, $\!\!13$, $\!\!24$, $\!\!34$, $\!\!123$, $\!\!124$, $\!\!134$, $\!\!234$, $\!\!1234\}$, $\!\{12$, $\!\!13$, $\!\!14$, $\!\!23$, $\!\!24$, $\!\!34$, $\!\!123$, $\!\!124$, $\!\!134$, $\!\!234\}$, $\!\{12$, $\!\!13$, $\!\!14$, $\!\!23$, $\!\!24$, $\!\!34$, $\!\!123$, $\!\!124$, $\!\!134$, $\!\!1234\}$, $\!\{12$, $\!\!13$, $\!\!14$, $\!\!23$, $\!\!24$, $\!\!123$, $\!\!124$, $\!\!134$, $\!\!234$, $\!\!1234\}$, $\!\{12$, $\!\!13$, $\!\!14$, $\!\!23$, $\!\!24$, $\!\!34$, $\!\!123$, $\!\!124$, $\!\!134$, $\!\!234$, $\!\!1234\}$
\end{scriptsize}


\bigskip 
\bigskip
\bigskip

\noindent
{\bf Acknowledgments.}
Authors thank Felipe Serrano for valuable technical support with {\tt SCIP}
 \bigskip \bigskip

\renewcommand*{\bibfont}{\small}
\printbibliography

\end{document}